\documentclass[journal,comsoc]{IEEEtran}
\usepackage[T1]{fontenc}

\usepackage{subfigure}
\usepackage{tabularx}
\usepackage{xcolor}
\usepackage{multirow}
\usepackage{booktabs}
\usepackage{array}
\newcolumntype{L}[1]{>{\raggedright\let\newline\\\arraybackslash\hspace{0pt}}m{#1}}
\newcolumntype{C}[1]{>{\centering\let\newline\\\arraybackslash\hspace{0pt}}m{#1}}
\newcolumntype{R}[1]{>{\raggedleft\let\newline\\\arraybackslash\hspace{0pt}}m{#1}}

\usepackage[english]{babel}
\usepackage{amsthm}
\usepackage{amssymb}

\newtheorem{lemma}{Lemma}
\newtheorem{corollary}{Corollary}
\newtheorem{proposition}{Proposition}
\newtheorem{definition}{Definition}

\newcommand{\E}{\mathbb{E}}
\newcommand{\C}{\mathbb{C}}

\newcommand{\cS}{\mathcal{S}}
\newcommand{\cT}{\mathcal{T}}

\newcommand{\dmin}{d_{\min}}
\newcommand{\kurt}{\mu_4}
\newcommand{\ism}{\nu_{-2}}
\newcommand{\MSE}{\mathrm{MSE}}
\newcommand{\CRB}{\mathrm{CRB}}
\newcommand{\SNR}{\mathrm{SNR}}
\newcommand{\SCR}{\mathrm{SCR}}
\newcommand{\diag}{\mathrm{diag}}

\usepackage{cite}

\ifCLASSINFOpdf
   \usepackage[pdftex]{graphicx}
\else
   or other class option (dvipsone, dvipdf, if not using dvips). graphicx
   \usepackage[dvips]{graphicx}
   \graphicspath{{../eps/}}
   \DeclareGraphicsExtensions{.eps}
\fi
\usepackage[normalem]{ulem}

\usepackage{amsmath}
\usepackage[cmintegrals]{newtxmath}
\usepackage{multicol}
\usepackage{algorithm}
\usepackage{algpseudocode}

\usepackage{verbatim}

\usepackage{array}
\usepackage{makecell}

\usepackage[pagewise]{lineno}
\usepackage{upgreek}

\ifCLASSOPTIONcompsoc
  \usepackage[caption=false,font=normalsize,uabelfont=sf,textfont=sf]{subfig}
\else
  \usepackage[caption=false,font=footnotesize]{subfig}
\fi
\begin{document}
\abovedisplayshortskip=1pt
\belowdisplayshortskip=1pt
\abovedisplayskip=1pt
\belowdisplayskip=1pt
\textfloatsep=1pt
\floatsep=1pt

\setcounter{figure}{0}
\renewcommand{\figurename}{Fig.}

\title{Physical-Layer Sensing Privacy via Constellation Shaping for OFDM-ISAC Systems: \\ Theory, Design, and Experiments}

\author{Kawon~Han,~\IEEEmembership{Member,~IEEE,}
        Kaitao~Meng,~\IEEEmembership{Member,~IEEE,}
        and~Christos~Masouros,~\IEEEmembership{Fellow,~IEEE}%
\thanks{K. Han is with the Department of Electrical Engineering, Ulsan National Institute of Science and Technology (UNIST), Ulsan, South Korea (e-mail: kawon.han@unist.ac.kr).}%
\thanks{K. Meng is with the Department of Electrical and Electronic Engineering, The University of Manchester, Manchester, M13 9PL, U.K. (e-mail: kaitao.meng@manchester.ac.uk).}%
\thanks{C. Masouros is with the Department of Electronic and Electrical Engineering, University College London, London, WC1E 6BT, U.K. (e-mail: c.masouros@ucl.ac.uk).}}

\maketitle
\begin{abstract}
The integration of sensing into communication networks introduces a new privacy risk, as a passive eavesdropper (Eve) may exploit ISAC data signals as signals of opportunity to perform unauthorized sensing of targets. In this paper, we develop a sensing-privacy-enhancing geometric constellation shaping (GCS) framework for OFDM-ISAC systems. The key observation is that constellation-dependent ranging performance is receiver-specific. For matched filtering at Eve, the ranging MSE is governed by the constellation kurtosis $\kurt$, whereas reciprocal filtering at the legitimate receiver (Alice) is governed by the inverse second-order moment $\ism$. Based on closed-form MSE expressions, we define sensing privacy as the ranging MSE gap between Eve and Alice and characterize its dependence on these two moments. The analysis shows that positive skewness of the symbol-power distribution is necessary for a positive intrinsic moment gap, namely $\kurt-\ism$. We further derive an exact skewness-based decomposition of the intrinsic moment gap and a canonical two-ring characterization, providing analytical guidelines for privacy-enhancing constellation geometries. We then formulate Eve-aware and Eve-agnostic GCS designs that balance sensing privacy and communication reliability through the minimum Euclidean distance (MED), with the Eve-agnostic design depending only on the intrinsic moment gap. Numerical results demonstrate scalable privacy--communication trade-offs, while over-the-air experiments show that the proposed constellation shaping substantially increases the ranging error gap between Eve and Alice with only a small communication throughput loss.
\end{abstract}

\begin{IEEEkeywords}
Constellation shaping, integrated sensing and communication (ISAC), orthogonal frequency division multiplexing (OFDM), physical layer security (PLS), sensing privacy.
\end{IEEEkeywords}

\IEEEpeerreviewmaketitle

\section{Introduction}

\IEEEPARstart{I}ntegrated sensing and communication (ISAC) is widely regarded as a key technology for sixth-generation (6G) wireless networks, enabling sensing and data transmission over a shared waveform, spectrum, and hardware platform \cite{Liu2022JSAC,Hassanien2016}. By reusing communication signals for sensing, ISAC improves spectral and hardware efficiency and supports applications ranging from autonomous driving and smart cities to industrial automation \cite{Zhang2021JSTSP,Liu2020TCOM,Meng2025WCM, han2025network}. In particular, communication-centric ISAC, where random data payloads modulated onto orthogonal frequency division multiplexing (OFDM) subcarriers are simultaneously exploited for sensing, provides a standard-compatible pathway toward ISAC deployment \cite{Sturm2011,LiuCPOFDM2025,Keskin2025}.

However, this dual use of communication signals introduces a distinct privacy concern. Target-reflected ISAC signals can propagate beyond the intended sensing receiver, allowing a passive radar eavesdropper (Eve) to exploit them as signals of opportunity for unauthorized target detection, localization, and tracking \cite{Berger2010,Qu2024}. Unlike communication data, which can be protected through encryption \cite{Shiu2011,Wang2018Survey}, sensing information is embedded directly in the physical propagation channel and may reveal privacy-sensitive information such as vital signs, speech-induced vibrations, and gestures \cite{Shah2019,Wan2014}. Protecting sensing privacy therefore requires physical-layer mechanisms that degrade Eve's sensing capability while preserving legitimate sensing and communication performance.

\subsection{Related Works}
Early studies on wireless sensing privacy primarily considered channel state information (CSI)-based sensing in WiFi systems. One class of approaches modifies the propagation environment using rotating antennas, reconfigurable intelligent surfaces (RIS), or dedicated reflectors to distort the channel observed by an unauthorized sensor \cite{Yao2024,Ruan2024,Staat2022,Shenoy2022}. Another approach randomizes the transmitted spatial signature through antenna scheduling or randomized beamforming \cite{Hernandez2023,Cominelli2024}. Pilot-obfuscation methods instead encode or scramble the reference signals so that only authorized receivers can recover the true CSI \cite{Ghiro2022,Abanto2020,Wang2024TIFS,Luo2024,Hu2024WiShield}. These approaches are effective when sensing relies on known pilots or CSI estimates. However, a passive radar Eve equipped with a reference channel can directly acquire an over-the-air copy of the transmitted waveform and correlate it with the target-reflected signal \cite{Berger2010}. In this case, hiding or modifying the nominal pilot structure alone does not prevent unauthorized sensing.

Sensing privacy has more recently been studied in ISAC systems through transceiver- and propagation-domain designs. Artificial-noise schemes inject controlled interference that can be canceled or exploited by the legitimate transceiver while degrading Eve's sensing signal-to-interference-plus-noise ratio (SINR) or sensing mutual information \cite{Zou2024,Musallam2025,Jia2025}. RIS-assisted designs create controlled reflections to suppress or distort Eve's observations \cite{Magbool2025}, while scatterer-assisted transmission deliberately illuminates environmental objects to generate additional clutter at Eve \cite{Chen2025}. Other approaches conceal the transmit directionality \cite{Ma2025}. Although these methods provide effective sensing-security mechanisms, many rely on knowledge of Eve's channel, location, or statistical geometry, and often consume additional transmit power, spatial degrees of freedom, or auxiliary hardware \cite{HanProcIEEE2026}. Moreover, the theoretical foundations of sensing privacy remain underdeveloped, and a systematic understanding of how physical-layer design parameters govern Eve's sensing capability is still largely unexplored.

At the waveform level, the work in \cite{Han2025SensingSecure} introduced Eve-agnostic sensing security through ambiguity-function engineering. A structured OFDM subcarrier power allocation creates deterministic ambiguity peaks in Eve's MF range profile, while Alice suppresses these peaks through RF by exploiting knowledge of the transmitted waveform. This work established the receiver knowledge asymmetry between Alice and Eve as a useful physical-layer resource for sensing security. The resulting security gain, however, comes at the cost of RF noise enhancement at Alice, with the trade-off later generalized through a Kullback--Leibler divergence (KLD)-based formulation in \cite{du2025securing}. In radar-centric ISAC, \cite{Temiz2026} similarly employs index and phase modulation of FMCW chirps to distort the ambiguity function perceived by unauthorized receivers, and also the exploitation of phase noise for ISAC sensing privacy is presented in \cite{keskin2026exploiting}.

Beyond these waveform-domain sensing-security approaches, recent studies have established the modulation constellation itself as an important design degree of freedom (DoF) in communication-centric ISAC. For MF sensing with random OFDM payloads, the average data-dependent sidelobe power is governed by the fourth-order moment, or kurtosis, of the constellation \cite{LiuIceberg2025}, motivating probabilistic and geometric constellation shaping for improved sensing performance \cite{Du2024,Yang2024,Geiger2025b,Hu2025,Xu2024ICUS, meng2026constellation}. For RF sensing, data equalization removes the random sidelobes but amplifies noise according to the inverse second-order moment of the constellation \cite{Wojaczek2019,Rodriguez2023,Mercier2020}. The work in \cite{Han2026OFDMISAC} unified these receiver-dependent effects through closed-form multi-target ranging MSE analysis. However, existing constellation-shaping studies have focused exclusively on improving the sensing performance of the intended receiver. Exploiting the distinct constellation dependencies of MF and RF receivers as a mechanism for sensing privacy remains unexplored.

\subsection{Motivation and Contributions}
The above results reveal a modulation-domain opportunity for protecting sensing privacy. Under the passive-radar model considered in this work, Eve does not know the random data payload a priori and therefore relies on MF processing using an over-the-air reference, whereas Alice knows the transmitted waveform and can apply a mismatched filter. The resulting ranging errors depend on different constellation moments: Eve's multi-scatterer MSE contains a clutter-induced term proportional to the constellation kurtosis $\kurt$, while Alice's MSE is governed by the inverse second-order moment $\ism$ through RF noise enhancement. Hence, constellation shaping can deliberately increase Eve's ranging error relative to Alice's without changing the OFDM subcarrier allocation or requiring additional spatial resources.

Based on this receiver asymmetry, we develop a sensing-privacy-enhancing GCS framework for OFDM-ISAC systems. \textbf{The main contributions are summarized as follows:}
\begin{itemize}

\item \textbf{Closed-form sensing privacy.}
We derive ranging MSE expressions for a passive Eve and a legitimate receiver in a multi-scatterer OFDM sensing environment, including Eve's imperfect reference signal. Defining sensing privacy as the ranging MSE gap between Eve and Alice yields a closed-form expression that separates Eve's clutter-dependent kurtosis penalty, the target-SNR difference, and Alice's RF noise-enhancement penalty.

\item \textbf{Fundamental characterization of privacy-enhancing constellation geometry.}
We characterize how the symbol-power distribution controls sensing privacy. For the Eve-agnostic moment gap $\kurt-\ism$, we derive an exact decomposition into a skewness term and a nonnegative residual, showing that positive power skewness is necessary for a positive moment gap. We further analyze the local behavior around PSK and a canonical two-ring family to reveal the resulting privacy-enhancing geometry.

\item \textbf{Eve-aware and Eve-agnostic geometric constellation shaping.}
We formulate GCS designs that balance sensing security and communication reliability. The Eve-aware design exploits statistical information about Eve's sensing environment, whereas the Eve-agnostic design depends only on the intrinsic moment gap $\kurt-\ism$ and requires no knowledge of Eve's location, channel, or clutter condition. The resulting non-convex problems are solved offline using a multi-start quasi-Newton search.

\item \textbf{Over-the-air experimental validation.}
An over-the-air SDR experiment validates the Eve-agnostic design in an uncontrolled outdoor environment. Despite Eve having a $6.7$~dB target-SNR advantage, the proposed GCS increases the measured ranging-deviation gap from $1.29$~m to $1.98$~m with less than $1\%$ communication-throughput loss.
\end{itemize}

\noindent \textit{Notations:} Boldface lower- and upper-case symbols denote vectors and matrices, respectively. $\C$ denotes the set of complex numbers. $(\cdot)^T$, $(\cdot)^H$, and $(\cdot)^*$ denote transpose, Hermitian transpose, and complex conjugation, respectively. $\diag(\mathbf{a})$ denotes a diagonal matrix with diagonal entries $\mathbf{a}$. The operators $\odot$ and $\oslash$ denote element-wise multiplication and division, respectively, and $\E[\cdot]$ denotes statistical expectation.

\section{System Model}
We consider an OFDM-based ISAC network illustrated in Fig.~\ref{fig:system}, comprising a single-antenna ISAC transmitter (TX), a legitimate sensing receiver (Alice), which may be collocated with the TX for monostatic sensing or spatially separated for bistatic sensing with prior TX signal information, a communication user, and an unauthorized passive sensing eavesdropper (Eve). The TX radiates an OFDM waveform with random data payloads to simultaneously serve the communication user and illuminate the targets of interest. The target-reflected signal reaches both Alice and Eve. In addition, a direct (surveillance) copy of the transmitted waveform leaks to Eve, providing it with a reference for passive-radar sensing. In the absence of any knowledge about the Eve, neither the reflection toward Eve nor the direct leakage can be prevented by the TX.

\begin{figure}[t!]
\centering
\includegraphics[width=0.8\columnwidth]{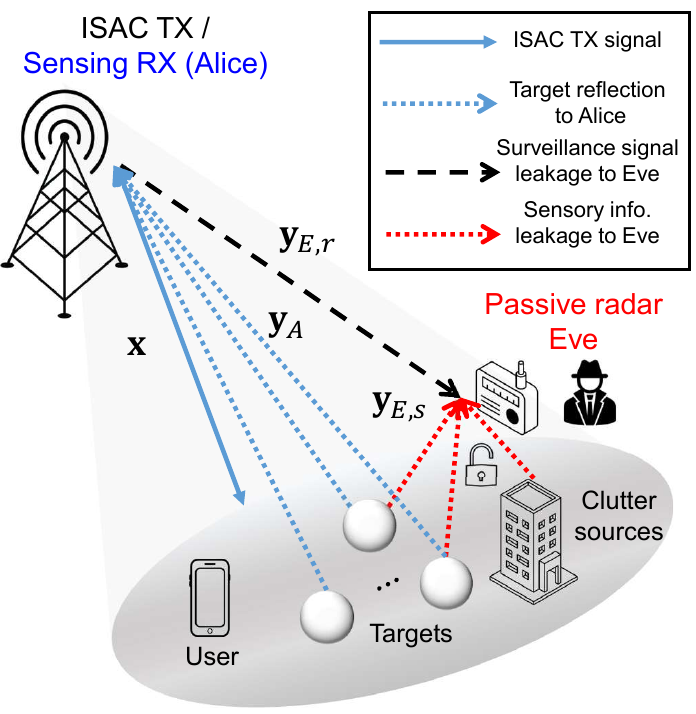}
\caption{Sensing-privacy scenario in OFDM-ISAC: the TX signal $\mathbf{x}$ supports communication and legitimate sensing at Alice, while surveillance-signal leakage $\mathbf{y}_{E,r}$ and sensory-information leakage $\mathbf{y}_{E,s}$ enable passive sensing at Eve, which is silent without transmitting any signals.}
\label{fig:system}
\end{figure}

To consider the practical ISAC deployment, we impose the following operational assumptions on Eve:
\begin{itemize}
\item \textbf{(A.1)} Eve operates as a passive bistatic radar equipped with a large receiving array or separately deployed directional antennas \cite{Wojaczek2019}, enabling the separation of the reference/direct-path (surveillance-signal leakage) signal from the target echo.
\item \textbf{(A.2)} Eve knows the location of the ISAC TX, i.e., its range and angle relative to Eve's position, and hence the line-of-sight (LoS) component of the reference channel.
\item \textbf{(A.3)} Eve has no prior knowledge of the transmitted ISAC signal, which is randomized by the communication payload.
\item \textbf{(A.4)} Eve has no prior range-Doppler information about the targets.
\item \textbf{(A.5)} For the Eve-agnostic case, the ISAC TX has no information about Eve, including its existence, location, or channel.
\end{itemize}
Assumption (A.5) describes the practically relevant worst case under (A.1)--(A.4), in which Eve remains completely silent and unknown, while the legitimate system must guarantee sensing privacy without any Eve-side information. 

\subsection{Transmit Signal Model}
The ISAC TX employs an OFDM waveform with $N$ subcarriers and subcarrier spacing $\Delta f = B/N$, where $B$ is the signal bandwidth. The frequency-domain transmit vector is
\begin{equation}
\mathbf{x} = [x_0, x_1, \ldots, x_{N-1}]^T, \qquad x_n \in \cS,
\end{equation}
where the communication symbols $x_n$ are drawn independently and uniformly from an $M$-ary constellation $\cS = \{s_1,\ldots,s_M\}$. Without loss of generality, the constellation is normalized to zero mean and unit average power, $\E[|x_n|^2]=1$. The statistical properties of $\cS$ that govern the ISAC performance under the conventional sensing receiver are given by \cite{Han2026OFDMISAC}
\begin{align}\label{eq:moments}
    \E\left[|x_n|^4\right] & = \kurt, \quad \forall x_n\in\cS \\
    \E\left[|x_n|^{-2}\right] & = \ism, \quad \forall x_n\in\cS,
\end{align}
where $\kurt = \frac{1}{M}\sum_{m=1}^{M}|s_m|^4$ is the fourth-order moment (kurtosis) and $\ism = \frac{1}{M}\sum_{m=1}^{M}|s_m|^{-2}$ is the inverse second-order moment of the constellation. Both moments equal unity for unit-amplitude constellations such as $M$-PSK, and exceed unity for non-unit-amplitude one. Representative values are $(\kurt,\ism)=(1.32,1.89)$ for 16QAM and $(1.38,2.69)$ for 64QAM \cite{Han2026OFDMISAC}. A cyclic prefix (CP) of length $N_{\rm cp}$ is prepended to each OFDM symbol. All target delays are assumed to lie within the inter-symbol interference (ISI)-free region determined by the CP \cite{Xu2025CP}, so that the CP can be omitted from the formulation after its removal at the receivers. The transmitted signal is written compactly in matrix form as $\mathbf{X}=\diag(\mathbf{x})$. In contrast to \cite{Han2025SensingSecure}, we deliberately employ equal subcarrier power allocation, so that all privacy enhancement is attributed to the constellation geometry alone. The combination with structured power allocation is discussed in Section~\ref{sec:conclusion}.

\subsection{Communication System Model}
After CP removal and discrete Fourier transform (DFT), the received frequency-domain signal at the communication user is expressed as
\begin{equation}
\mathbf{y}_c = \mathbf{H}_c\,\mathbf{x} + \mathbf{z}_c,
\end{equation}
where $\mathbf{H}_c = \diag(h_0,\ldots,h_{N-1})$ collects the per-subcarrier channel gains of a frequency-selective fading channel, and $\mathbf{z}_c\sim\mathcal{CN}(\mathbf{0},\sigma_c^2\mathbf{I}_N)$ is additive white Gaussian noise (AWGN). At high SNR, the uncoded reliability of the link is dominated by the pairwise confusion of the two closest constellation points \cite{Caire2002}. Accordingly, we adopt the minimum Euclidean distance (MED) as the communication performance metric for constellation design, which is written as
\begin{equation}\label{eq:MED}
\dmin = \min_{s_i\neq s_j\in\cS} |s_i - s_j|.
\end{equation}
The end-to-end communication performance is further quantified by the measured throughput
\begin{equation}\label{eq:throughput}
\eta~\text{(bits/s)} = \log_2 M \cdot (1-\mathrm{BLER})\cdot\frac{N}{T_{\rm sym}+T_{\rm cp}},
\end{equation}
where $T_{\rm sym}$ and $T_{\rm cp}$ denote the OFDM symbol and CP durations, respectively, and $\mathrm{BLER}$ denotes the block error rate determined as the fraction of transmitted data blocks that are unsuccessfully decoded.

\subsection{Sensing Signal Model at the Legitimate Receiver (Alice)}
Let $K_A$ discrete scatterers, encompassing both targets of interest and environmental clutter sources, be located at distinct, resolvable ranges in Alice's field of view. Scatterers are assumed static or slowly moving so that Doppler-induced inter-carrier interference (ICI) is negligible \cite{Keskin2021ICI}. After CP removal and DFT, the frequency-domain received signal at Alice is given by
\begin{equation}\label{eq:alice_rx}
\mathbf{y}_A = \mathbf{a}_A^T \mathbf{H}_A \mathbf{X} + \mathbf{z}_A,
\end{equation}
where $\mathbf{a}_A = [\alpha_{A,1},\ldots,\alpha_{A,K_A}]^T\in\C^{K_A\times 1}$ collects the complex amplitudes including path loss and radar cross-section (RCS), $\mathbf{H}_A = [\mathbf{h}(\tau_{A,1}),\ldots,\mathbf{h}(\tau_{A,K_A})]^T\in\C^{K_A\times N}$ is the delay channel matrix with range steering vectors
\begin{equation}
\mathbf{h}(\tau) = \left[1, e^{-j2\pi\Delta f\tau},\ldots,e^{-j2\pi(N-1)\Delta f\tau}\right]^T \in \C^{N\times 1},
\end{equation}
where $\tau_{A,k}$ is the round-trip time of flight (TOF) of scatterer $k$, and $\mathbf{z}_A\sim\mathcal{CN}(\mathbf{0},\sigma_A^2\mathbf{I}_N)$.

Since Alice has full knowledge of $\mathbf{X}$, it is possible to utilize mismatched filtering (MMF) to suppress the sidelobes induced by the random payload. To maximize the sidelobe suppression, Alice employs reciprocal filtering (RF) \cite{Wojaczek2019,Rodriguez2023}, i.e., element-wise equalization of the random payload, which is given by
\begin{equation}\label{eq:alice_rf}
\mathbf{y}_{A,\rm RF} = \mathbf{y}_A \oslash \mathbf{x} = \mathbf{a}_A^T\mathbf{H}_A + \mathbf{z}_{A,\rm RF},
\end{equation}
where the equalized noise $\mathbf{z}_{A,\rm RF}=\mathbf{z}_A\oslash\mathbf{x}$ remains zero mean with reshaped per-subcarrier variance as
\begin{equation}\label{eq:rf_noise}
\mathrm{Var}(z_{A,{\rm RF},n}) = \sigma_A^2\,\E\!\left[|x_n|^{-2}\right] = \sigma_A^2\,\ism .
\end{equation}
Hence, RF removes the data-dependent modulation of the sensing channel at the cost of an SNR loss proportional to $\ism$. This is the fundamental mechanism by which the constellation geometry shapes Alice's performance.

\subsection{Sensing Signal Model at the Eavesdropper (Eve)}
Eve is a passive bistatic radar that observes a surveillance (target-echo) signal and a reference signal leaked directly from the TX. Under (A.1), these are separated at Eve's front-end. The surveillance signal reflected by $K_E$ scatterers including targets and clutter is expressed as
\begin{equation}\label{eq:eve_rx}
\mathbf{y}_{E,s} = \mathbf{a}_E^T \mathbf{H}_E \mathbf{X} + \mathbf{z}_{E,s},
\end{equation}
with $\mathbf{a}_E=[\alpha_{E,1},\ldots,\alpha_{E,K_E}]^T$, $\mathbf{H}_E=[\mathbf{h}(\tau_{E,1}),\ldots,\mathbf{h}(\tau_{E,K_E})]^T$, bistatic TOFs $\tau_{E,k}$, and $\mathbf{z}_{E,s}\sim\mathcal{CN}(\mathbf{0},\sigma_E^2\mathbf{I}_N)$. The reference signal follows the Rician channel model, which is given by
\begin{align}\label{eq:eve_ref}
    \mathbf{y}_{E,r} & = \mathbf{h}_{E,r}\odot\mathbf{x} + \mathbf{z}_{E,r},\quad \\
    \mathbf{h}_{E,r} & = \sqrt{\tfrac{g_{E,r}\mathcal{K}}{\mathcal{K}+1}}\,\mathbf{h}_{E,\rm LoS} + \sqrt{\tfrac{g_{E,r}}{\mathcal{K}+1}}\,\mathbf{h}_{E,\rm NLoS},
\end{align}
where $g_{E,r}$ is the reference channel gain, $\mathcal{K}$ the Rician factor, $\mathbf{h}_{E,\rm LoS}$ the LoS component known to Eve, and $\mathbf{h}_{E,\rm NLoS}\sim\mathcal{CN}(\mathbf{0},\mathbf{I}_N)$. Compensating the known LoS path, Eve obtains the effective reference channel as
\begin{equation}\label{eq:eve_ref_eff}
\tilde{\mathbf{y}}_{E,r} = \sqrt{\tfrac{g_{E,r}\mathcal{K}}{\mathcal{K}+1}}\,\mathbf{x} + \sqrt{\tfrac{g_{E,r}}{\mathcal{K}+1}}\,\tilde{\mathbf{h}}_{E,\rm NLoS}\odot\mathbf{x} + \tilde{\mathbf{z}}_{E,r}.
\end{equation}
It should be noted that the effective reference at Eve is a noisy and multipath-corrupted estimate of the transmitted waveform, whose quality is quantified by the reference signal-to-interference-plus-noise ratio (SINR) given by $\gamma_r = \mathcal{K}g_{E,r}/(g_{E,r}+\sigma_E^2(\mathcal{K}+1))$.

\noindent \textbf{Remark 1.} Unlike Alice, Eve relies on a noisy, multipath-corrupted reference rather than exact knowledge of the transmitted payload. Direct RF processing using \eqref{eq:eve_ref_eff} can therefore amplify reference errors and incur substantial output-SNR
loss~\cite{Han2025SensingSecure}. We thus consider Eve employing reference-based matched filtering (MF), $\mathbf{y}_{E,\rm MF} = \mathbf{y}_{E,s}\odot\tilde{\mathbf{y}}_{E,r}^{*}$.\footnote{Cross-correlation between the reference and surveillance signals is a standard passive-radar approach that requires no payload decoding~\cite{liu2015performance}. This implements matched filtering using the received reference as the waveform template.} We do not claim MF to be universally optimal for Eve; rather, our analysis targets the practically relevant passive-radar regime, in which Eve lacks the transmitted payload and relies on an imperfect over-the-air reference. The underlying methodology could be adapted to alternative receiver architectures at Eve by incorporating the corresponding signal-processing models and re-deriving Eve’s sensing-performance characterization. Such extensions would require receiver-specific analysis to reassess the resulting sensing-privacy performance and are beyond the scope of this work. Table~\ref{tab:operational} summarizes this receiver configuration, which defines the scope of our privacy analysis. Whereas \cite{Han2025SensingSecure} exploited this MF--RF asymmetry through the AF of a power-shaped waveform, we exploit it through the constellation moments $\kurt$ and $\ism$, which govern Eve's and Alice's ranging MSEs, respectively.

\begin{table}[t]
\caption{Operational Differences Between Alice and Eve Exploited for Sensing Privacy}
\label{tab:operational}
\centering
\begin{tabular}{lcccc}
\toprule
 & Sensing & TX signal & Receiver & Governing \\
 & mode & knowledge & processing & moment \\
\midrule
Alice & Monostatic & Fully-known & RF & $\ism$ \\
Eve & Bistatic (passive) & Unknown & MF & $\kurt$ \\
\bottomrule
\end{tabular}
\end{table}

\section{Receiver-Specific Ranging Performance: Preliminaries}\label{sec:analysis}
The estimation-theoretic framework for OFDM ranging over random data payloads was established in \cite{Han2026OFDMISAC}. Here, we first recall the two receiver-specific results required by the privacy analysis, restated in the adversarial bistatic setting of Section~II.

As a benchmark, the expected CRB for delay estimation under random signaling is given by \cite{Han2026OFDMISAC}
\begin{equation}\label{eq:ecrb}
\E[\CRB_{\tau_k}] \approx \frac{\sigma^2}{8\pi^2\Delta f^2 |\alpha_k|^2}\left(\frac{3}{N^3} + \frac{27(\kurt-1)}{5N^4}\right),
\end{equation}
in which constellation-dependent term decays as the number of subcarriers $N$ grows and is negligible for a practical number of subcarriers. The fundamental limit is thus essentially constellation-invariant. Sensing privacy cannot be engineered at the level of the estimation-theoretic bound, but through the gap between the practical receivers available to Alice and Eve, whose achievable MSEs depend on the constellation in fundamentally different ways.

Both receivers estimate delays by standard dictionary-based peak search over their respective filter outputs,
\begin{align}\label{eq:eve_est}
\hat{\tau}_{E,k} & = \arg\max_{\tau\in\cT}\big|\mathbf{h}^H(\tau)\,\mathbf{y}_{E,\rm MF}^T\big|,\quad \\ \label{eq:alice_est}
\hat{\tau}_{A,k} & = \arg\max_{\tau\in\cT}\big|\mathbf{h}^H(\tau)\,\mathbf{y}_{A,\rm RF}^T\big|,
\end{align}
over a delay dictionary $\cT$ of sufficient resolution, possibly refined by subspace methods. Now, the following two Lemmas specify the ranging MSEs of Eve and Alice, respectively.

\begin{lemma}[Eve's ranging MSE]\label{thm:eve_mse}
Under MF processing with the estimator \eqref{eq:eve_est}, in the high-SNR regime and for resolvable scatterers, the delay estimation MSE of the $k$-th scatterer at Eve satisfies
\begin{equation}\label{eq:mse_eve}
\MSE_{{\rm Eve},k}(\cS) \geq \frac{3\left((\kurt-1)\sum_{j\neq k}^{K_E}|\alpha_{E,j}|^2 + \sigma_E^2\right)}{8\pi^2\Delta f^2\,|\alpha_{E,k}|^2\,N^3},
\end{equation}
with equality in the ideal-reference SINR $\gamma_r\to\infty$.
\end{lemma}
\begin{IEEEproof}
See Appendix~\ref{app:mse}.
\end{IEEEproof}

\begin{lemma}[Alice's ranging MSE]\label{thm:alice_mse}
Under RF processing with the estimator \eqref{eq:alice_est}, in the high-SNR regime and for resolvable scatterers, the delay estimation MSE of the $k$-th scatterer at Alice is
\begin{equation}\label{eq:mse_alice}
\MSE_{{\rm Alice},k}(\cS) = \frac{3\,\sigma_A^2\,\ism}{8\pi^2\Delta f^2\,|\alpha_{A,k}|^2\,N^3}.
\end{equation}
\end{lemma}
\begin{IEEEproof}
See Appendix~\ref{app:mse}.
\end{IEEEproof}

\noindent \textbf{Remark 2.} Lemmas~\ref{thm:eve_mse} and \ref{thm:alice_mse} clarify the asymmetry of Table~\ref{tab:operational}. Eve's accuracy is limited by a multi-target or clutter interference floor $(\kurt-1)\sum_{j\neq k}|\alpha_{E,j}|^2$, seeded by every other scatterer in its field of view and scaled by the constellation kurtosis. Moreover, since the reference-induced terms in \eqref{eq:eve_ref_eff} are uncontrollable by the TX and strictly harmful to Eve, \eqref{eq:mse_eve} is a lower bound, so any privacy guarantee built on it is conservative. On the other hand, Alice is immune to the scene. RF equalization converts all data-dependent sidelobes into a uniform noise amplification governed solely by $\ism$, at an SNR cost of $10\log_{10}\ism$~dB relative to the PSK baseline. The design task that follows is therefore to grow $\kurt$ aggressively while keeping $\ism$ small, which is a moment decoupling that only the constellation geometry can perform.

\section{Sensing Privacy: Definition and Properties}\label{sec:privacy}
In this section, we develop the core theory of sensing privacy at the modulation level. Specifically, we first define the privacy metric and characterize it in closed form, then establish its fundamental properties and invariances, derive the Eve-aware and Eve-agnostic design objectives, and finally analyze a canonical constellation family that reveals the geometry of privacy-optimal signaling.

\subsection{Sensing Privacy Metric and Closed-Form Characterization}
To quantify the level of protection for sensing privacy, we define the sensing privacy as the ranging MSE gap between the eavesdropper and the legitimate receiver, as follows.
\begin{definition}[Sensing privacy]\label{def:privacy}
For a target of interest indexed by $k$ at both receivers, the sensing privacy of the constellation $\cS$ is defined as the ranging MSE gap between Eve and Alice, given by
\begin{equation}\label{eq:privacy_def}
\Delta\varepsilon_k(\cS) \triangleq \MSE_{{\rm Eve},k}(\cS) - \MSE_{{\rm Alice},k}(\cS).
\end{equation}
\end{definition}
A large positive $\Delta\varepsilon_k$ indicates that Eve's range estimates are substantially less accurate than Alice's, i.e., the target's location is effectively kept private. Substituting \eqref{eq:mse_eve} and \eqref{eq:mse_alice} into \eqref{eq:privacy_def} and normalizing by the common factor $3/(8\pi^2\Delta f^2 N^3)$, we obtain the lower-bound of the normalized MSE gap, given by
\begin{align}
\Delta\bar{\varepsilon}_k(\cS)
&= (\kurt-1)\underbrace{\frac{\sum_{j\neq k}^{K_E}|\alpha_{E,j}|^2}{|\alpha_{E,k}|^2}}_{c_0\,=\,1/\SCR_E}
+ \underbrace{\frac{\sigma_E^2}{|\alpha_{E,k}|^2}}_{c_1\,=\,1/\SNR_E}
- \ism\underbrace{\frac{\sigma_A^2}{|\alpha_{A,k}|^2}}_{c_2\,=\,1/\SNR_A} \nonumber\\
&= (\kurt-1)\,c_0 + c_1 - \ism\,c_2, \label{eq:privacy_closed}
\end{align}
where $\SCR_E$ denotes the signal-to-clutter ratio (SCR) of the target of interest at Eve, and $\SNR_E$ and $\SNR_A$ denote the target SNRs at Eve and Alice, respectively. Expression \eqref{eq:privacy_closed} decouples the role of the propagation environment, which enters only through the scalars $(c_0,c_1,c_2)$, from that of the constellation, which enters only through the moment pair $(\kurt,\ism)$.

\begin{figure*}[t]
\centering
\includegraphics[width=1.8\columnwidth]{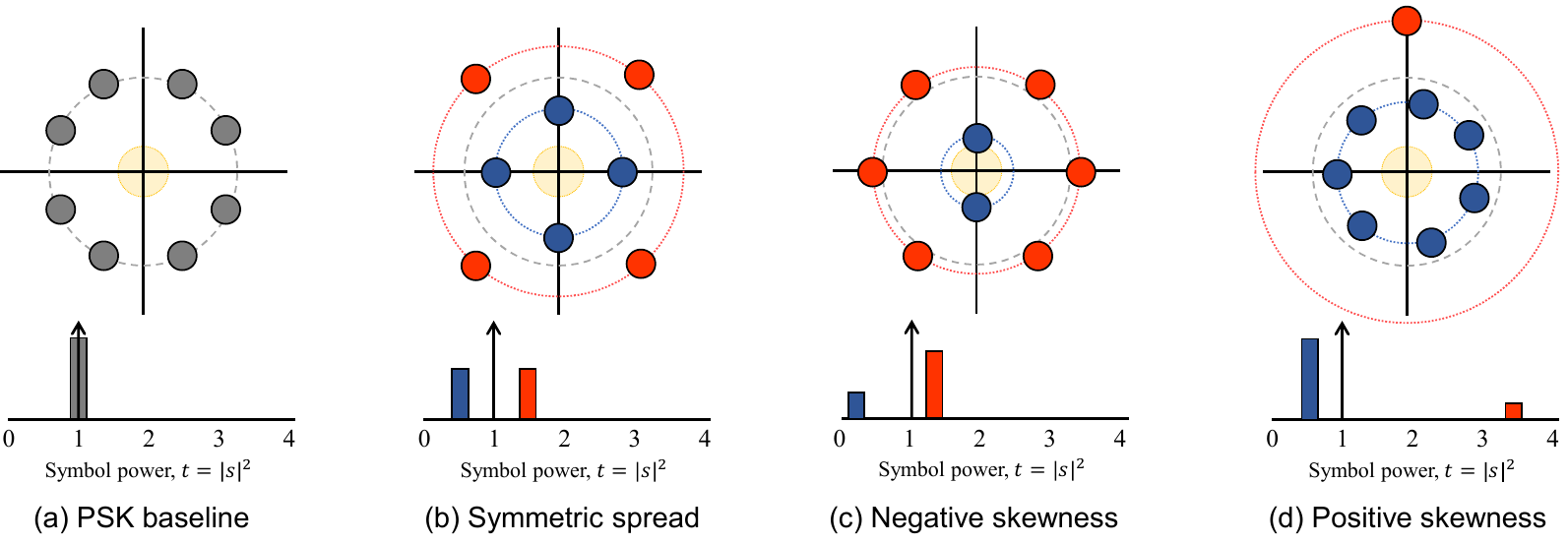}
\caption{Illustration of constellations and their symbol-power distributions relevant to sensing privacy: (a) PSK baseline, (b) symmetric power spread, (c) negative power skewness, and (d) positive power skewness.}
\label{fig:skew}
\end{figure*}

Building on the above definition, we make three structural observations as follows:
\begin{enumerate}
\item The effect of the constellation kurtosis $\kurt$ is amplified by the Eve's clutter environment through $c_0$: the denser the scattering scene around the protected target from Eve's viewpoint, the more privacy each unit of kurtosis buys. Cluttered urban and indoor scenes, exactly where sensing privacy matters most, are thus intrinsically favorable for protecting sensing privacy.
\item The term $c_1$ is independent of $\cS$ and represents a fixed privacy offset due to Eve's thermal noise.
\item The cost term $\ism c_2$ is paid at Alice regardless of Eve's environment. It is suppressed by a strong legitimate link, i.e., a small $c_2$, which can be secured by Alice's coherent processing gain over multiple OFDM symbols.
\end{enumerate}
An immediate consequence of \eqref{eq:privacy_closed} is the positive-privacy condition when $c_0 \neq 0$, which is given by
\begin{equation}\label{eq:pos_privacy}
\Delta\bar{\varepsilon}_k > 0 \;\Longleftrightarrow\; \kurt - 1 > \frac{\ism\,c_2 - c_1}{c_0},
\end{equation}
which quantifies how much kurtosis must be purchased, per unit of noise enhancement $\ism$, to overcome an adverse link-budget gap $c_2>c_1$, in which Eve enjoys a better raw link than Alice. It is worth noting that the defined sensing privacy represents the worst case for Alice, or equivalently the best case for Eve, since the SINR degradation of Eve's reference link further deteriorates Eve's sensing performance beyond \eqref{eq:mse_eve}, as described in Lemma~\ref{thm:eve_mse}.

\subsection{Fundamental Properties of Sensing Privacy}
Building on the closed-form characterization of sensing privacy in \eqref{eq:privacy_closed}, we investigate the level of privacy achievable with conventional constellations and identify which departures from them yield meaningful gains. We address these questions through a series of properties, beginning with the observation that the two governing moments are lower-bounded by unity.
\begin{lemma}\label{lem:jensen}
For any zero-mean, unit-power constellation, it holds that $\kurt\geq 1$ and $\ism\geq 1$, with equality in both if and only if $|s_m|=1$ for all $m$.
\end{lemma}
\begin{IEEEproof}
Let $t=|x_n|^2$. Since the constellation has unit average power, we have $\E[t]=\E[|x_n|^2]=1$. First, because $t^2$ is convex, Jensen's inequality gives $\kurt=\E[t^2]\geq \E[t]^2=1.$ Similarly, since $t^{-1}$ is convex for $t>0$, $\ism=\E[t^{-1}]\geq \frac{1}{\E[t]}=1$. Equality in either Jensen inequality holds only when $t$ is constant. Together with $\E[t]=1$, this requires $t=1$, or equivalently, $|s_m|=1$ for all constellation points.
\end{IEEEproof}
Lemma~\ref{lem:jensen} shows that any amplitude variation in the constellation has competing effects: it increases Eve's data-dependent sidelobe floor through $\kurt$, while also increasing Alice's noise enhancement through $\ism$. The key design question is therefore whether the privacy gain from the former can outweigh the performance loss caused by the latter. Unit-amplitude signaling is the unique case in which both effects are minimized, and its privacy performance is characterized as follows.

\begin{corollary}\label{cor:psk}
For any unit-amplitude constellation, $\Delta\bar{\varepsilon}_k = c_1 - c_2$. Thus, the sensing privacy reduces to the link-budget gap between Eve and Alice and cannot be further improved through constellation design. In particular, if Eve has a comparable or more favorable sensing geometry than Alice ($c_1\leq c_2$), PSK signaling provides no sensing privacy.
\end{corollary}

Corollary~\ref{cor:psk} highlights an important limitation of conventional PSK signaling. Although PSK is sensing-optimal for both MF and RF receivers \cite{LiuIceberg2025,Han2026OFDMISAC}, it provides no inherent sensing privacy, since Eve benefits from the same sidelobe-free range profile as Alice. Improving sensing privacy therefore requires a departure from unit-amplitude signaling. To determine which amplitude variations are beneficial, we examine the Eve-agnostic moment gap $\kurt-\ism$ for small power fluctuations around the unit-amplitude point, leading to the following result.

\begin{proposition}\label{prop:skewness}
Let $t=|x_n|^2$ denote the instantaneous symbol power with $\E[t]=1$, and define the power fluctuation as $u=t-1$. Then, the intrinsic moment gap admits the exact decomposition
\begin{equation}\label{eq:exact_gap}
\kurt - \ism = \E\left[u^3\right] - \E\!\left[\frac{u^4}{t}\right],
\end{equation}
where the first term represents the unnormalized skewness of the symbol-power distribution, while the second term is a nonnegative residual. Consequently,
\begin{equation}\label{eq:skew_necessary}
\kurt - \ism \leq \E\left[u^3\right],
\end{equation}
which shows that positive power skewness, $\E[u^3]>0$, is a necessary condition for achieving a positive moment gap. Furthermore, when the fourth- and higher-order fluctuations are negligible, the residual term in \eqref{eq:exact_gap} becomes negligible, and the moment gap is well approximated by the skewness alone as $\kurt - \ism \approx \E\left[u^3\right].$
\end{proposition}
\begin{IEEEproof}
See Appendix~\ref{app:skewness}.
\end{IEEEproof}

Proposition~\ref{prop:skewness} provides a useful insight into the role of symbol-power shaping in sensing privacy: the moment gap is governed not simply by the spread of the symbol power, but by its asymmetry. The two terms in \eqref{eq:exact_gap} capture competing effects. The skewness term favors a positively skewed power distribution, in which a small fraction of high-power symbols contributes disproportionately to Eve's data-dependent sidelobes through $\kurt$, while the more frequent, moderately lower-power symbols incur a comparatively smaller noise-enhancement penalty at Alice. In contrast, the residual term $\E[u^4/t]$ always reduces the moment gap. Because it is inversely weighted by the instantaneous symbol power, this penalty becomes particularly large when constellation points approach the origin or when the power distribution deviates substantially from unit-amplitude signaling. Fig.~\ref{fig:skew} illustrates these different power profiles, ranging from unit-amplitude and symmetric spreading to negatively and positively skewed symbol-power distributions.

Now, the role of the power spread itself is clarified by the following corollary.
\begin{corollary}\label{cor:flatness}
Since $\kurt=1+\E[u^2]$ and $\ism=1+\E[u^2]-\E[u^3]+\E[u^4/t]$ hold exactly, substituting them into \eqref{eq:privacy_closed} yields the exact decomposition
\begin{equation}\label{eq:flatness}
\Delta\bar{\varepsilon}_k(\cS)
=
(c_1-c_2)
+(c_0-c_2)\,\E[u^2]
+c_2\,\E[u^3]
-c_2\,\E\!\left[\frac{u^4}{t}\right].
\end{equation}
\end{corollary}
\begin{IEEEproof}
The moment expressions are established in Appendix~\ref{app:skewness}, and substituting them into \eqref{eq:privacy_closed} yields \eqref{eq:flatness}.
\end{IEEEproof}
Corollary~\ref{cor:flatness} shows that simply spreading the symbol amplitudes around the PSK point as shown in Fig.~\ref{fig:skew}(b) does not necessarily improve sensing privacy. For small power fluctuations, $\E[u^2]$, $\E[u^3]$, and $\E[u^4/t]$ contribute at the second, third, and fourth orders, respectively. The leading effect is therefore the power-spread term $(c_0-c_2)\E[u^2]$, whose sign depends on the sensing environment. When $c_0=c_2$, this second-order contribution cancels exactly, and the leading privacy variation is governed by the power asymmetry $\E[u^3]$. Thus, symmetric amplitude spreading provides no second-order privacy gain, whereas a positively skewed power distribution as shown in Fig.~\ref{fig:skew}(d), with a small number of high-power symbols and a larger number of moderately lower-power symbols, can improve privacy through the third-order term. This benefit is ultimately limited by the residual $c_2\E[u^4/t]$, which becomes increasingly unfavorable as the power variation grows or constellation points approach the origin. More generally, power spread is beneficial when $c_0>c_2$, while for $c_0<c_2$, sufficiently positive power skewness is required to overcome the unfavorable second-order effect.

\subsection{Eve-Aware and Eve-Agnostic Privacy Objectives}\label{sec:objectives}
Depending on the information available at the TX, we consider two secure constellation-design regimes.
\subsubsection{Eve-aware design}
The Eve-aware design assumes statistical or bounded knowledge of Eve's sensing operating conditions. Specifically, the TX knows or estimates the large-scale \textbf{target and clutter powers observed at Eve, together with Eve's effective receiver-noise power}, or equivalently, the expected SCR and sensing SNR required to determine the coefficients $c_0$, $c_1$, and $c_2$. Such information may be obtained from prior observations or specified over a prescribed worst-case eavesdropping region.

Under this knowledge, the TX can directly optimize the sensing-privacy metric in \eqref{eq:privacy_closed}. Since the term involving $c_1$ is independent of the constellation, the effective Eve-aware objective is
\begin{equation}\label{eq:obj_aware}
f_{\rm aw}(\cS) = (\kurt-1)c_0-\ism c_2.
\end{equation}
Thus, the optimal constellation explicitly adapts to Eve's statistical sensing environment. In particular, it balances the increase in Eve's data-dependent sensing error achieved through a larger $\kurt$ against the corresponding noise-enhancement penalty at Alice, characterized by $\ism$.

\subsubsection{Eve-agnostic design}
Under the practical assumption (A.5), neither $c_0$ nor $c_1$ is available at the TX. To obtain a constellation-design criterion that does not depend on Eve-specific information, we adopt the symmetric reference condition
\begin{equation}\label{eq:agnostic_cond}
c \triangleq c_0=c_1=c_2.
\end{equation}
This represents a neutral reference point, at which Eve's SCR and SNR are equal to Alice's SNR. Importantly, (27) is not an assumption on Eve's actual operating condition, but an Eve-agnostic reference used to isolate a purely constellation-dependent privacy criterion. Proposition 2 subsequently identifies sufficient conditions under which this reference provides a conservative privacy guarantee. 

Substituting \eqref{eq:agnostic_cond} into \eqref{eq:privacy_closed} yields
\begin{equation}
\Delta\bar{\varepsilon}_k=c\left(\kurt-\ism\right).
\end{equation}
Since $c>0$ is independent of the constellation, the Eve-agnostic objective reduces to the moment gap
\begin{equation}\label{eq:obj_agnostic}
f_{\rm ag}(\cS)=\kurt-\ism.
\end{equation}
Unlike the Eve-aware objective, $f_{\rm ag}(\cS)$ depends only on the constellation and requires no knowledge of Eve's sensing geometry or propagation environment. It therefore provides a simple criterion for practical constellation design when Eve-specific information is unavailable.

The following proposition establishes a condition, under which Eve-agnostic design provides a conservative sensing privacy guarantee.
\begin{proposition}\label{prop:conservative}
Let $\cS$ be any constellation satisfying $\kurt-\ism\geq0$. If the sensing environment satisfies $c_0\geq c_2, \;c_1\geq c_2,$ then the realized sensing privacy satisfies
\begin{equation}\label{eq:conservative}
\Delta\bar{\varepsilon}_k
\geq
c_2\left(\kurt-\ism\right)
\geq
0.
\end{equation}
\end{proposition}

\begin{IEEEproof}
From Lemma~\ref{lem:jensen}, $\kurt\geq1$. Therefore, under $c_0\geq c_2$,
\begin{equation}
(\kurt-1)c_0
\geq
(\kurt-1)c_2.
\end{equation}
Together with $c_1\geq c_2$, substituting these inequalities into \eqref{eq:privacy_closed} gives
\begin{align}
\Delta\bar{\varepsilon}_k & =
(\kurt-1)c_0+c_1-\ism c_2 \\
&\geq (\kurt-1)c_2+c_2-\ism c_2 =
c_2\left(\kurt-\ism\right).
\end{align}
The final inequality in \eqref{eq:conservative} follows from $\kurt-\ism\geq0$.
\end{IEEEproof}

\noindent \textbf{Remark 3.} Proposition~\ref{prop:conservative} shows that the symmetric reference condition is conservative whenever the combination of the sensing environment and Eve/Alice geometry satisfies $c_0\geq c_2$ and $c_1\geq c_2$. Under these conditions, any constellation with a nonnegative moment gap guarantees nonnegative sensing privacy, and the realized privacy is no smaller than the Eve-agnostic reference value. Such conditions may arise when Eve experiences stronger clutter, a less favorable sensing geometry, or additional reference-induced degradation. The above insights allow us to form principles of constellation design for sensing privacy without knowledge of Eve. Below, we employ these to design constellations for sensing privacy.

\vspace{-2mm}
\subsection{Canonical Two-Ring Characterization}\label{sec:tworing}
\begin{figure}[t]
\centering
\includegraphics[width=0.65\columnwidth]{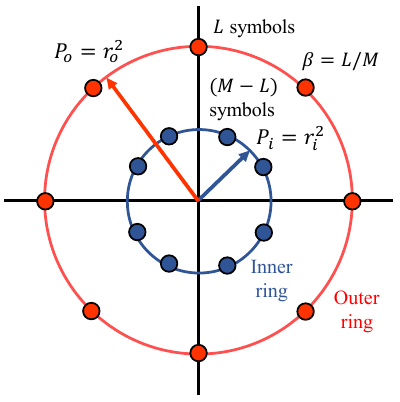}
\caption{Canonical two-ring $M$-ary constellation model: $L$ symbols are placed on an outer ring with power $P_o=r_o^2$, while the remaining $M-L$ symbols lie on an inner ring with power $P_i=r_i^2$, with outer-ring fraction $\beta=L/M$.}
\label{fig:tworing}
\end{figure}
Building on the moment gap $\kurt-\ism$ as the Eve-agnostic design objective, we now examine its geometry within a canonical two-ring constellation. This family captures the key trade-off identified above: positive power skewness can increase the moment gap, whereas excessively low-power symbols incur a large $\ism$ penalty. Specifically, let $L$ symbols lie on an outer ring with power $P_o=r_o^2$ and the remaining $M-L$ symbols on an inner ring with power $P_i=r_i^2$. The resulting two-ring geometry is illustrated in Fig.~\ref{fig:tworing}, where a fraction $\beta=L/M$ of the symbols occupies the outer ring and the remaining symbols lie on the inner ring. Then, the unit-power constraint is given by
\begin{equation}\label{eq:tworing_power}
\beta P_o + (1-\beta) P_i = 1, \qquad 0<P_i<1<P_o .
\end{equation}
The corresponding moments are
\begin{align}\label{eq:tworing_moments}
\kurt & = \beta P_o^2 + (1-\beta)P_i^2,\\
\ism & = \frac{\beta}{P_o} + \frac{1-\beta}{P_i}.
\end{align}

Let $G(\beta,P_o)=\kurt-\ism$ denote the Eve-agnostic moment gap. Using \eqref{eq:tworing_power}, the inner-ring power can be written as $P_i=(1-\beta P_o)/(1-\beta)$. The constraint $0<P_i<1<P_o$ then gives the feasible range $P_o\in(1,1/\beta)$, over which $G$ becomes a scalar function of $P_o$. Its behavior is characterized as follows.

\begin{proposition}\label{prop:tworing}
For an outer-ring fraction $\beta\in(0,1)$, the moment gap satisfies
\begin{equation}\label{eq:tworing_grad}
\frac{\partial G}{\partial P_o} = \beta\,(P_o-P_i)\left[2 - \frac{P_o+P_i}{P_o^2 P_i^2}\right],\quad P_o\in\left(1,\tfrac{1}{\beta}\right),
\end{equation}
from which the following properties follow:
\begin{enumerate}
\item[(i)] Near the PSK point $P_o=P_i=1$, letting $u_o=P_o-1$, the moment gap behaves as
\begin{equation}\label{eq:tworing_onset}
G \approx \beta\left(1-\frac{\beta^2}{(1-\beta)^2}\right)u_o^3 .
\end{equation}
Thus, a small two-ring departure increases the moment gap if and only if $\beta<1/2$.

\item[(ii)] As $P_o\to 1/\beta$, or equivalently $P_i\to0$, the moment gap decreases without bound, i.e., $G\to-\infty$.

\item[(iii)] For $\beta<1/2$, $G$ has a unique stationary point in the feasible interval, which is its global maximum. The optimal powers satisfy
\begin{equation}\label{eq:tworing_opt}
2\,(P_o^\star)^2 (P_i^\star)^2 = P_o^\star + P_i^\star,\qquad P_i^\star=\frac{1-\beta P_o^\star}{1-\beta}.
\end{equation}
\end{enumerate}
\end{proposition}
\begin{IEEEproof}
See Appendix~\ref{app:tworing}.
\end{IEEEproof}

Proposition~\ref{prop:tworing} gives a simple geometric interpretation of the preceding results. Property (i) shows that a beneficial two-ring design requires fewer symbols on the high-power outer ring than on the inner ring ($\beta<1/2$), which is consistent with the positive power skewness identified in Proposition~\ref{prop:skewness}. However, increasing the ring separation indefinitely is not beneficial. As shown in property (ii), pushing the inner ring toward the origin causes the $\ism$ penalty to dominate and drives the moment gap to $-\infty$. These two effects lead to the finite optimum in property (iii), which balances a small number of high-power outer symbols against sufficiently powered inner symbols. For example, with $(M,L)=(8,1)$, \eqref{eq:tworing_opt} gives $P_o^\star\approx5.91$ and $P_i^\star\approx0.30$, yielding $\kurt\approx4.45$, $\ism\approx2.96$, and $G^\star\approx1.49$.

\noindent \textbf{Remark 4.} The finite optimum in Proposition~\ref{prop:tworing} holds for a fixed outer-ring fraction $\beta$. If $\beta$ is also allowed to approach zero, the Eve-agnostic moment gap becomes unbounded. Specifically, let $\Delta=\beta(P_o-1)$ be fixed while $\beta\to0$. Then $P_o=1+\Delta/\beta\to\infty$, while $P_i=1-\Delta/(1-\beta)$ remains finite for $0<\Delta<1$. Consequently, $\kurt$ grows as $\Delta^2/\beta$, whereas $\ism$ remains bounded, and hence $G=\kurt-\ism\to\infty$. This limit corresponds to increasingly rare, high-power symbols and therefore leads to excessive peak-to-average power ratio (PAPR) and poor communication performance. Communication constraints, such as the MED constraint, are therefore essential for obtaining a practically meaningful sensing-privacy-optimal constellation.

\section{Geometric Constellation Shaping for Sensing Privacy Protection}\label{sec:design}
We now incorporate the privacy objectives \eqref{eq:obj_aware} and \eqref{eq:obj_agnostic} into a geometric constellation shaping (GCS) framework that balances sensing privacy and communication reliability.

\subsection{Problem Formulation}
The design variables are the symbol coordinates $\cS=\{s_m\}_{m=1}^M\subset\C$. Following \cite{Han2026OFDMISAC}, communication reliability is characterized by the MED in \eqref{eq:MED}, while zero-mean, zero-pseudo-variance, and unit-power constraints are imposed to maintain the basic constellation structure. The Eve-aware design problem is formulated as
\begin{subequations}\label{eq:P1}
\begin{align}
\text{(P.1)}\quad \max_{\{s_m\}_{m=1}^M}\;\;
& \rho\,\frac{f_{\rm aw}(\cS)}{f_{\rm aw}^{\max}}
+ (1-\rho)\,\frac{\dmin}{\dmin^{\max}}
\label{eq:P1obj}\\
\text{s.t.}\quad
& |s_i-s_j| \geq \dmin,\quad \forall\,s_i\neq s_j\in\cS,
\label{eq:P1med}\\
& \sum_{m=1}^{M}s_m=0,\quad
  \sum_{m=1}^{M}s_m^2=0,\quad
  \frac{1}{M}\sum_{m=1}^{M}|s_m|^2=1.
\label{eq:P1stat}
\end{align}
\end{subequations}
where $\rho\in[0,1]$ controls the privacy--communication trade-off. The normalization constants $f_{\rm aw}^{\max}$ and $\dmin^{\max}$ are the corresponding single-objective values at $\rho=1$ and $\rho=0$, respectively. The three constraints in \eqref{eq:P1stat} impose zero mean, zero pseudo-variance, and unit average power, respectively.

The Eve-agnostic problem (P.2) follows by replacing $f_{\rm aw}$ with $f_{\rm ag}=\kurt-\ism$:
\begin{equation}\label{eq:P2}
\text{(P.2)}\quad
\max_{\{s_m\}_{m=1}^M}\;
\rho\,\frac{\kurt-\ism}{f_{\rm ag}^{\max}}
+(1-\rho)\,\frac{\dmin}{\dmin^{\max}}
\quad
\text{s.t.}\;\eqref{eq:P1med},\,\eqref{eq:P1stat}.
\end{equation}
At $\rho=0$, the formulation reduces to conventional MED-maximizing constellation shaping, whereas increasing $\rho$ progressively places more emphasis on sensing privacy. At $\rho=1$ \footnote{Since the privacy-only design at $\rho=1$ may favor excessively large outer-ring amplitudes, we use $\rho=0.99$ in practice to retain a small MED contribution and avoid overly peaky constellation geometries.}, the design is determined solely by the corresponding privacy objective for the given constellation size $M$.

Here, sensing privacy favors increasing $\kurt$ to degrade Eve's sensing performance while keeping $\ism$ sufficiently small to limit the noise enhancement at Alice. The resulting optimization therefore promotes the positively skewed symbol-power distributions identified in Proposition~\ref{prop:skewness}, while the MED objective preserves communication reliability.

\subsection{Solution via Reformulation and Multi-Start Search}

Problems (P.1) and (P.2) are non-convex due to the MED constraints, the equality constraints, and the nonlinear dependence of $\ism$ on the constellation coordinates. To simplify the numerical optimization, we eliminate the unit-power constraint through normalization and evaluate the MED directly from the constellation. Specifically, the constellation is rescaled to unit average power at every objective evaluation, while $\dmin(\cS)$ in \eqref{eq:MED} is computed directly rather than introduced as an auxiliary optimization variable.

Normalizing the Eve-aware privacy objective by $c_2$, we define $\gamma_0=c_0/c_2$. The resulting composite objective is
\begin{equation}\label{eq:composite_obj} 
F(\cS) = \rho\,\frac{(\kurt-1)\,\gamma_0-\ism}{\bar{f}} + (1-\rho)\,\frac{\dmin(\cS)}{\bar{d}}, 
\end{equation} 
where $\bar{f}$ and $\bar{d}$ are the normalization constants of \eqref{eq:P1obj}. The Eve-aware design therefore depends on the environment only through the relative coefficient $\gamma_0$, while the Eve-agnostic design is conducted by setting $\gamma_0=1$.

\begin{algorithm}[t] 
\caption{Multi-Start Search for Sensing-Secure GCS} 
\label{alg:gcs} 
    \begin{algorithmic}[1] 
    \Require Order $M$, weight $\rho$, ratio $\gamma_0$ ($\gamma_0{=}1$ for Eve-agnostic), number of starts $N_{\rm init}$ 
    \For{$i=1,\ldots,N_{\rm init}$} 
        \State Draw initial constellation $\cS^{(i)}_0$: random Gaussian points, $M$-QAM, and $M$-PSK
        \State Normalize $\cS^{(i)}_0$ to unit average power
        \State $\hat{\cS}^{(i)}\gets$ quasi-Newton maximization of $F$ in \eqref{eq:composite_obj} from $\cS^{(i)}_0$, with unit-power normalization applied at every evaluation
    \EndFor 
    \State \Return $\cS^{(i^\star)}$ with $i^\star=\arg\max_i F\big(\hat{\cS}^{(i)}\big)$ 
    \end{algorithmic} 
\end{algorithm} 

The reformulated problem is solved using a multi-start quasi-Newton search. Random initializations are supplemented with an $M$-QAM seed, and the solution achieving the largest objective value is retained, as summarized in Algorithm~\ref{alg:gcs}. No explicit minimum-amplitude constraint is imposed, since $\ism$ naturally penalizes constellation points approaching the origin. For the privacy-oriented design, we use $\rho=0.99$ rather than $\rho=1$ so that a small MED contribution is retained to discourage degenerate constellation geometries. The per-evaluation cost is dominated by the $\mathcal{O}(M^2)$ pairwise-distance calculations required to evaluate $\dmin(\cS)$. Since the optimization is performed offline, the resulting constellations can be stored as fixed mapping tables without additional run-time optimization.

\begin{figure}[t]
\centering
    \subfigure[]{\includegraphics[width=0.24\textwidth]{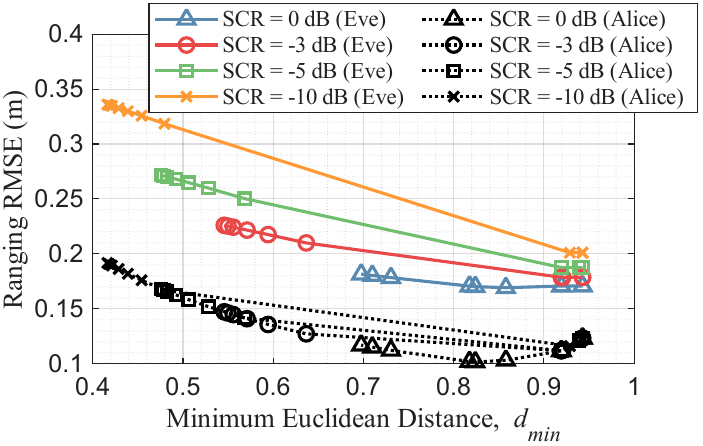}}
    \subfigure[]{\includegraphics[width=0.24\textwidth]{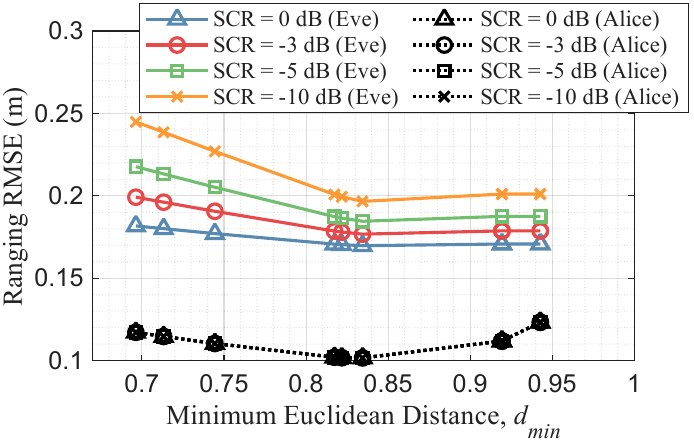}} 
    \caption{Ranging RMSEs of Eve and Alice under (a) Eve-aware GCS and (b) Eve-agnostic GCS.}
    \label{fig:R-RMSE}
\end{figure}

\begin{figure}[t]
    \centering
    \includegraphics[width=0.88\columnwidth]{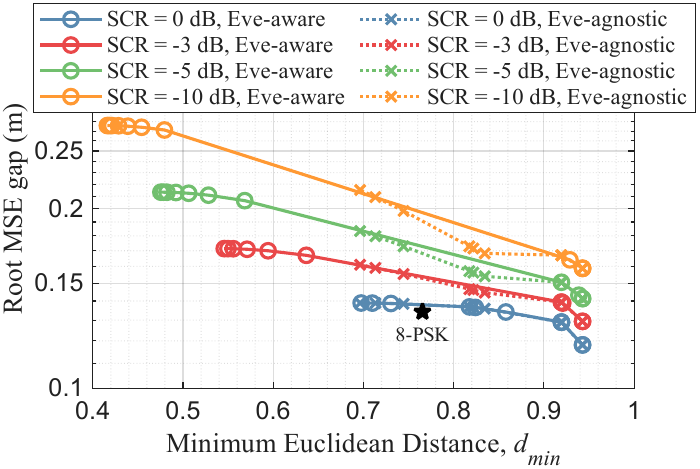}
    \caption{Privacy--communication trade-off: simulated ranging MSE gap between Eve and Alice versus the achieved MED for the proposed $8$-ary GCS designs, with $\SNR_A=0$~dB, $\SNR_E=-5$~dB, and $\SCR_E$ $\in\{0,-3,-5,-10\}$~dB}
\label{fig:tradeoff}
\end{figure}

\section{Simulation and Experimental Results}\label{sec:results}
In this section, the proposed framework is evaluated in two stages. Numerical simulations first quantify the privacy--communication trade-off and characterize the resulting constellation geometries under controlled conditions. An over-the-air experiment then validates the Eve-agnostic design in a practical propagation environment.

\subsection{Privacy--Communication Trade-off}
The simulation employs a CP-OFDM signal with $N=256$ subcarriers, equal subcarrier power allocation, and the designed $8$-ary constellations, so that the observed effects arise solely from the constellation geometry. The per-subcarrier target SNRs are set to $\SNR_A=1/c_2=0$~dB at Alice and $\SNR_E=1/c_1=-5$~dB at Eve, while Eve's SCR, $1/c_0$, is varied over $\{0,-3,-5,-10\}$~dB. Alice applies RF processing in \eqref{eq:alice_rf}, whereas Eve applies MF processing using its reference signal, followed by matrix-pencil-based range estimation. The constellations are obtained using Algorithm~\ref{alg:gcs} for both the Eve-aware design, with $\gamma_0$ matched to each SCR condition, and the Eve-agnostic design, with $\gamma_0=1$. Each result is averaged over $1000$ Monte-Carlo runs.

\begin{figure}[t]
\centering
    \subfigure[]{\includegraphics[width=0.20\textwidth]{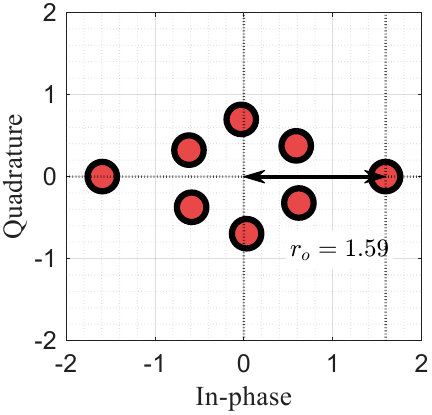}}
    \subfigure[]{\includegraphics[width=0.20\textwidth]{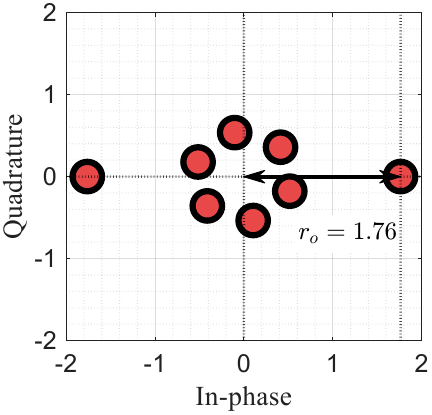}} \\
    \subfigure[]{\includegraphics[width=0.20\textwidth]{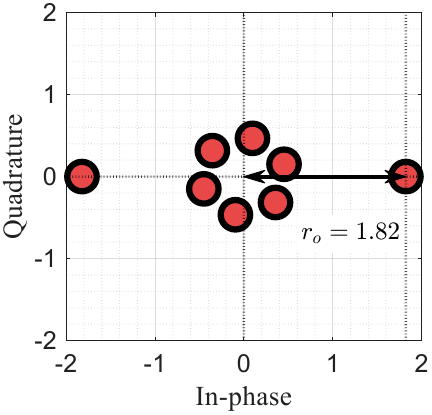}} 
    \subfigure[]{\includegraphics[width=0.20\textwidth]{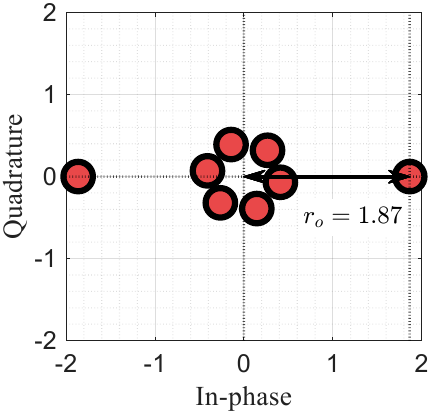}}
    \caption{Privacy-oriented Eve-aware constellations for $\SCR_E=$ (a) $0$~dB, (b) $-3$~dB, (c) $-5$~dB, and (d) $-10$~dB.}
    \label{fig:opt_const}
\end{figure}

\begin{table}[t]
\caption{Benchmark of Proposed and Standard Constellations Under the Default Budget ($\SNR_A=0$ dB, $\SNR_E=-5$ dB, $\SCR_E=-5$ dB)}
\label{tab:comparison}
\centering
\begin{tabular}{lccccc}
\toprule
Constellation & $\dmin$ & $\kurt$ & $\ism$  & Ranging gap \\
\midrule
$8$-PSK & $0.77$ & $1.00$ & $1.00$& 0.134 m \\
$16$-QAM & $0.63$ & $1.32$ & $1.89$  & 0.126 m \\
GCS, (P.2), $\rho=0.9$ & $0.74$ & $1.41$ & $1.22$  & 0.156 m \\
GCS, (P.1), $\rho=0.9$ & $0.56$ & $2.43$ & $2.51$  & 0.171 m \\
\bottomrule
\end{tabular}
\end{table}

Fig.~\ref{fig:R-RMSE} provides the ranging RMSE of Eve and Alice under the proposed GCS, and Fig.~\ref{fig:tradeoff} shows the ranging RMSE gap between Eve and Alice versus the achieved MED, where each curve is obtained by sweeping $\rho$ from $0$ to $1$. Three observations are evident. First, both designs exhibit a clear privacy--communication trade-off controlled by $\rho$. For example, at $\SCR_E=-5$~dB, the Eve-aware design increases the RMSE gap from approximately $0.14$~m to $0.21$~m as $\dmin$ decreases from $0.94$ to $0.48$, while at $\SCR_E=-10$~dB the gap reaches approximately $0.27$~m. Second, the Eve-aware design achieves a larger RMSE gap than the Eve-agnostic design because it exploits knowledge of $\gamma_0$. Nevertheless, the performance difference remains moderate. For example, the Eve-agnostic design achieves an RMSE gap of approximately $0.22$~m at $\SCR_E=-10$~dB without knowledge of Eve's clutter condition. Third, decreasing Eve's SCR shifts the trade-off frontier upward for both designs. This behavior is consistent with \eqref{eq:privacy_closed}: stronger clutter increases the contribution of the kurtosis-dependent sidelobe interference to Eve's sensing error, making the proposed privacy mechanism more effective.

\begin{figure}[t]
\centering
    \subfigure[]{\includegraphics[width=0.2\textwidth]{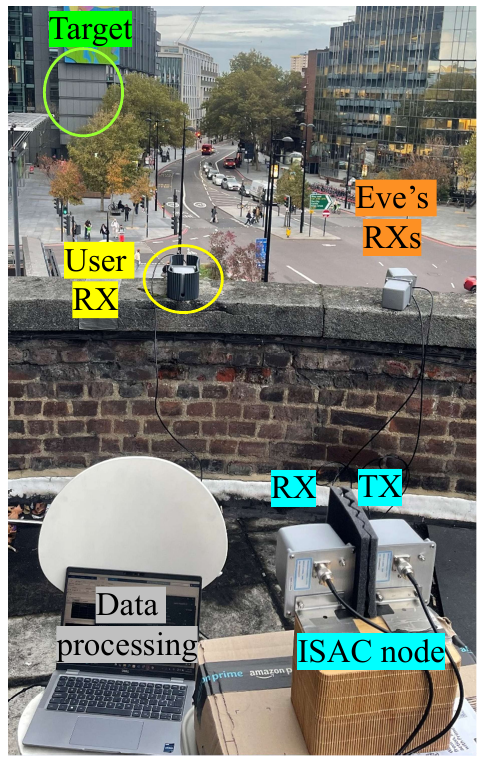}}
    \subfigure[]{\includegraphics[width=0.2\textwidth]{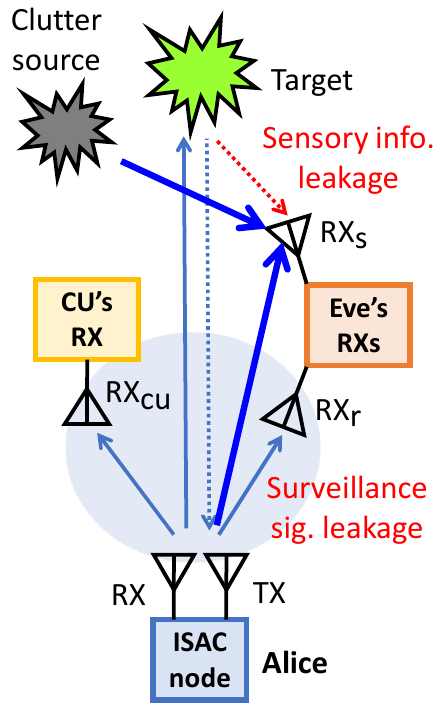}} \\
    \subfigure[]{\includegraphics[width=0.4\textwidth]{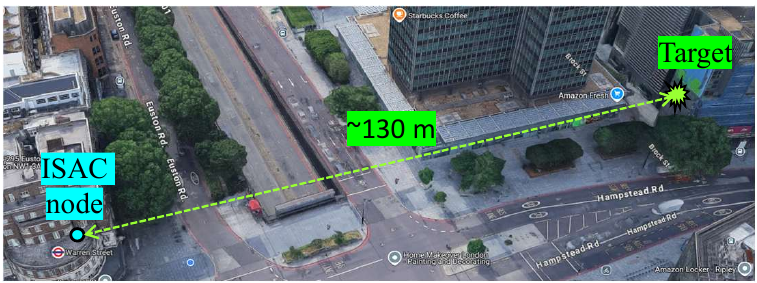}} 
    \caption{Over-the-air experimental setup: (a) photograph of the deployment showing the ISAC node, communication-user RX, and Eve's reference and surveillance RXs; (b) schematic of the deployment, including a dominant clutter scatterer within Eve's field of view, and (c) aerial view of the measurement site.}
\label{fig:setup}
\end{figure}

\subsection{Privacy-Oriented Constellation Geometries}
Fig.~\ref{fig:opt_const} shows the privacy-oriented ($\rho\to1$) Eve-aware constellations obtained under the four SCR conditions. A common feature across the optimized geometries is the emergence of two dominant high-power symbols, while the remaining six symbols form a lower-power cluster away from the origin. Notably, these inner symbols do not lie on a common ring, implying that the optimized constellations are not exact two-ring structures. Nevertheless, the concentration of power on a small subset of symbols is consistent with the positive power skewness identified in Proposition~\ref{prop:skewness} and with the qualitative insight of Proposition~\ref{prop:tworing}. As Eve's SCR decreases, the two dominant symbols move progressively farther from the origin, while the remaining symbols contract to maintain the unit-power constraint. This behavior reflects the increasing importance of the kurtosis term under stronger clutter, resulting in a more skewed power distribution at the expense of MED.

Table~\ref{tab:comparison} compares the proposed designs at $\rho=0.9$ with $8$-PSK and $16$-QAM under $\SCR_E=-5$~dB. The results clearly illustrate the different roles of the Eve-aware and Eve-agnostic objectives. As predicted by Corollary~\ref{cor:psk}, $8$-PSK has $\kurt=\ism=1$ and therefore provides only the baseline ranging gap of $0.134$~m. The $16$-QAM benchmark yields $\kurt=1.32$ and $\ism=1.89$, resulting in a negative moment gap and a slightly smaller ranging gap of $0.126$~m. In contrast, the Eve-agnostic design (P.2) achieves $\kurt-\ism=0.19$ and increases the ranging gap to $0.156$~m while retaining an MED of $0.74$, close to that of $8$-PSK. The Eve-aware design (P.1) further increases the ranging gap to $0.171$~m by more aggressively increasing $\kurt$, although this also raises $\ism$ and reduces the MED to $0.56$. Notably, its intrinsic moment gap is slightly negative, $\kurt-\ism=-0.08$, demonstrating that the Eve-aware design need not maximize the moment gap itself. Instead, it exploits the stronger weighting of the kurtosis term through the known clutter coefficient $c_0$.

\begin{table}[t]
\caption{OFDM-ISAC Prototype System Parameters}
\label{tab:prototype}
\centering
\begin{tabular}{lc}
\toprule
Specification & Value \\
\midrule
Center frequency, $f_c$ & $2.4$ GHz \\
Bandwidth, $B$ & $20$ MHz \\
Number of subcarriers, $N$ & $512$ \\
CP length & $1.6$ $\mu$s \\
Symbol duration & $14.4$ $\mu$s \\
Number of symbols per frame & $512$ \\
Sampling rate & $60$ MHz \\
Antenna gain & $12$ dBi \\
Modulation order, $M$ & $8$ \\
\bottomrule
\end{tabular}
\end{table}

\begin{figure}[t]
\centering
    \subfigure[]{\includegraphics[width=0.24\textwidth]{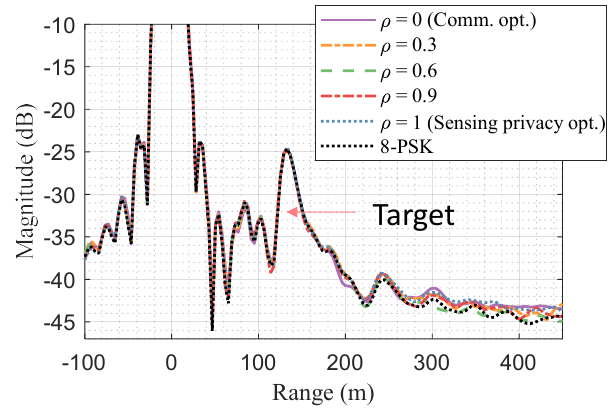}}
    \subfigure[]{\includegraphics[width=0.24\textwidth]{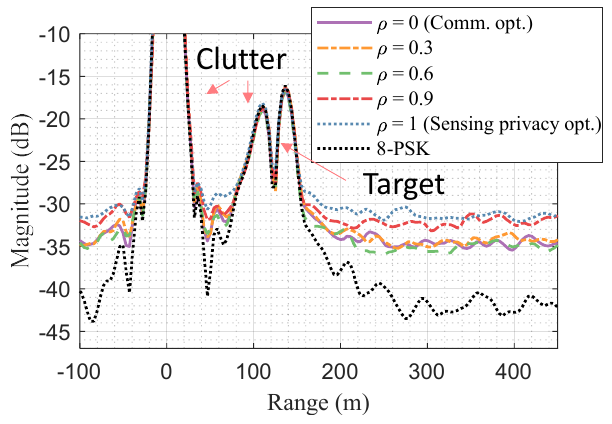}} 
    \caption{Measured range profiles for the Eve-agnostic GCS designs with $\rho\in\{0,0.3,0.6,0.9,1\}$ and the $8$-PSK baseline: (a) Alice with RF processing and (b) Eve with MF processing.}
\label{fig:range_profiles}
\end{figure}

\begin{figure*}[t]
\centering
\includegraphics[width=1.8\columnwidth]{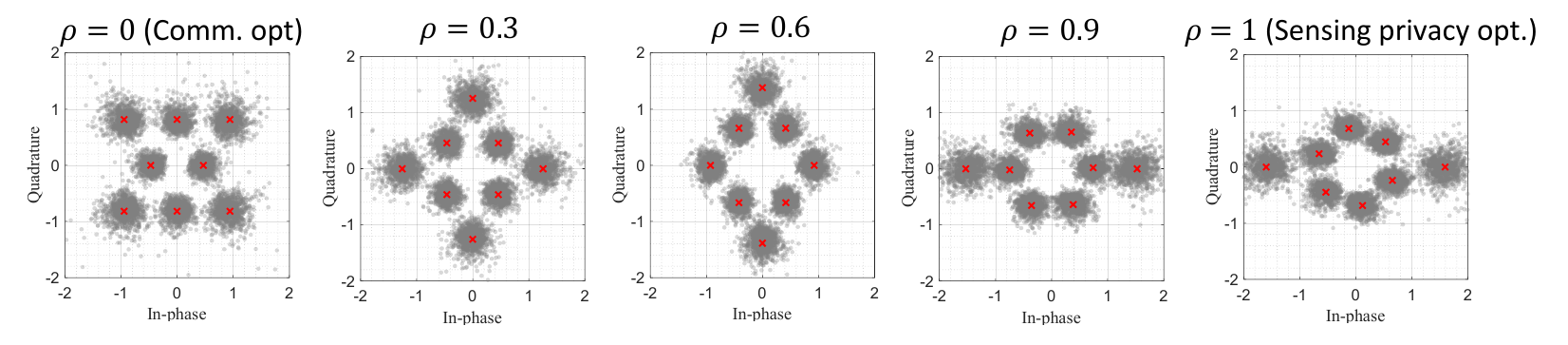}
\caption{Measured constellations at the communication user for the Eve-agnostic designs with $\rho\in\{0,0.3,0.6,0.9,1\}$, shown from left to right.}
\label{fig:meas_const}
\end{figure*}

\subsection{Over-the-air Experimental Results}
\subsubsection{Experimental Setup}
The over-the-air validation extends the monostatic ISAC measurement to an adversarial sensing scenario and focuses on the Eve-agnostic design. The prototype, shown in Fig.~\ref{fig:setup}(a), consists of an AD9363-based SDR connected to directional antennas with $12$~dBi gain and deployed on a rooftop overlooking an urban intersection. The ISAC node (Alice) contains the collocated TX and sensing RX, while a communication-user receiver ($\mathrm{RX}_{\rm cu}$) is located within the coverage area. As illustrated in Fig.~\ref{fig:setup}(b), Eve is emulated by two spatially separated receive chains in accordance with (A.1): a reference antenna $\mathrm{RX}_r$ directed toward the TX and a surveillance antenna $\mathrm{RX}_s$ directed toward the scene. The interested target is a reflective building located at approximately $130$~m from the ISAC node, as shown in Fig.~\ref{fig:setup}(c), and the scene also contains a dominant clutter scatterer within Eve's field of view, providing the multi-scatterer condition considered in Lemma~\ref{thm:eve_mse}. The measured per-subcarrier link parameters before coherent integration are $\SNR_E=1/c_1=-24.19$~dB, $\SNR_A=1/c_2=-30.90$~dB, and $\SCR_E=1/c_0=-28.60$~dB, so that Eve has a $6.7$~dB target-SNR advantage over Alice, providing a challenging sensing-privacy condition in which PSK signaling would favor Eve in terms of ranging accuracy. The OFDM parameters are summarized in Table~\ref{tab:prototype}: carrier frequency $2.4$~GHz, bandwidth $B=20$~MHz, $N=512$ subcarriers, and $512$ symbols per frame, yielding $27$~dB of coherent processing gain. All experimental results are averaged over $100$ independent frames.

\begin{figure}[t]
\centering
\includegraphics[width=0.8\columnwidth]{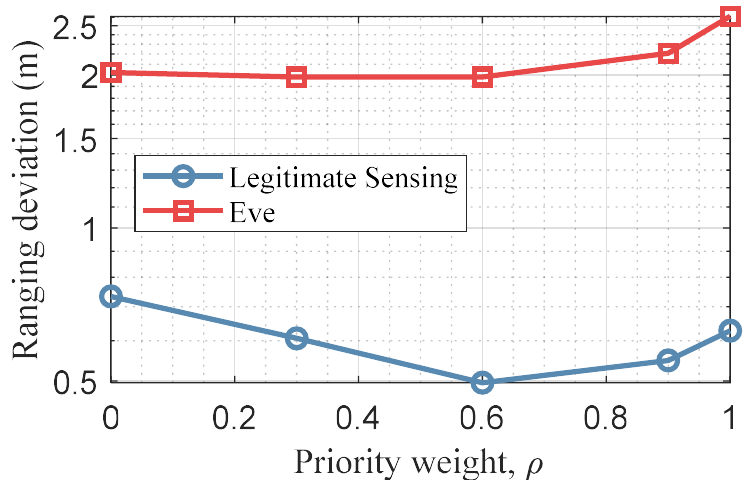}
\caption{Measured ranging deviation at Alice and Eve versus the priority weight $\rho$.}
\label{fig:deviation}
\end{figure}

\subsubsection{Measured Range Profiles and Constellations}
Fig.~\ref{fig:range_profiles} shows the measured range profiles at Alice and Eve. At Alice, as shown in Fig.~\ref{fig:range_profiles}(a), the target peak at approximately $130$~m remains clearly resolved for all designs, and the profiles vary only slightly with $\rho$. Relative to the $8$-PSK baseline, the noise floor increases by at most approximately $1.5$~dB at $\rho=1$, consistent with the $\ism$-dependent noise enhancement in \eqref{eq:rf_noise}. In contrast, Fig.~\ref{fig:range_profiles}(b) shows a pronounced increase in the sidelobe floor at Eve. The strong clutter return generates data-dependent sidelobes, as predicted by \eqref{eq:mse_eve}, raising the floor from approximately $-41$~dB for $8$-PSK to about $-35$~dB at $\rho=0$ and $-31$~dB at $\rho=1$. Thus, the proposed GCS increases Eve's interference floor by approximately $6$--$10$~dB while leaving Alice's range profile largely unchanged. This contrast experimentally demonstrates the receiver asymmetry underlying the proposed sensing-privacy framework. In particular, the relatively clean Eve profile obtained with $8$-PSK is consistent with Corollary~\ref{cor:psk}, which shows that unit-amplitude signaling provides no constellation-induced sensing privacy.

Fig.~\ref{fig:meas_const} shows the corresponding constellations measured at the communication user for the Eve-agnostic designs with $\rho\in\{0,0.3,0.6,0.9,1\}$. Despite channel distortion and noise, the measured constellations preserve the main geometric trend of the designed constellations: as $\rho$ increases, a subset of symbols moves toward higher amplitudes, producing an increasingly skewed symbol-power distribution.

\begin{figure}[t]
\centering
\includegraphics[width=0.8\columnwidth]{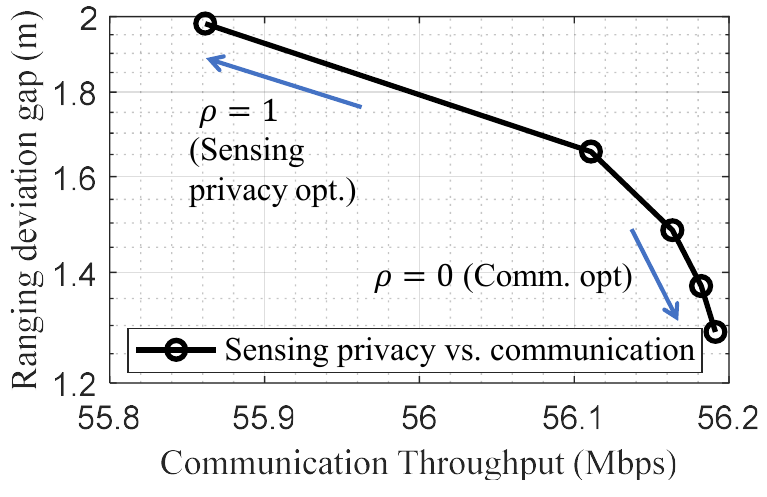}
\caption{Measured privacy--throughput trade-off for the Eve-agnostic design: ranging deviation gap (Eve $-$ Alice) versus communication throughput as $\rho$ varies from $0$ to $1$.}
\label{fig:gap_tput}
\end{figure}

\subsubsection{Ranging Deviation and Privacy--Throughput Trade-off}
The estimation performance over repeated frames is quantified by the ranging deviation
\begin{equation}\label{eq:deviation}
\delta = \sqrt{\frac{1}{N_f}\sum_{i=1}^{N_f}\Big(\hat{r}_i - \frac{1}{N_f}\sum_{\ell=1}^{N_f}\hat{r}_\ell\Big)^{2}},
\end{equation}
i.e., the empirical standard deviation of the per-frame range estimates $\{\hat r_i\}_{i=1}^{N_f}$, which removes the effect of the small static bias associated with the fixed bistatic geometry of Eve. Since the bistatic geometry is fixed across all constellation designs, the associated static bias is approximately constellation-independent. Hence, the ranging deviation serves as a bias-independent proxy for the constellation-dependent random-error component of the ranging MSE. Fig.~\ref{fig:deviation} shows $\delta$ at Alice and Eve as a function of $\rho$. Alice's deviation remains between $0.49$ and $0.73$~m over the entire sweep, with no clear monotonic dependence on $\rho$, indicating that the $\ism$-induced degradation remains small relative to the measurement variability. In contrast, Eve's deviation is approximately $2.0$~m for $\rho\leq0.6$ and increases to $2.20$~m at $\rho=0.9$ and $2.61$~m at $\rho=1$, approximately four times Alice's deviation of $0.63$~m at the same operating point. The relatively large deviation already observed at $\rho=0$ arises because the communication-oriented GCS is itself non-unit-amplitude, producing clutter-induced data-dependent sidelobes even without an explicit privacy weight.

Fig.~\ref{fig:gap_tput} summarizes the measured sensing-privacy--communication trade-off by plotting the ranging deviation gap between Eve and Alice against the communication throughput computed from \eqref{eq:throughput} using the measured BLER at the user RX. As $\rho$ increases from $0$ to $1$, the deviation gap increases from $1.29$~m to $1.98$~m, while the throughput decreases only from $56.19$ to $55.86$~Mbps, corresponding to a loss below $0.6\%$. These results demonstrate that the proposed modulation-level design can substantially increase Eve's ranging uncertainty with only a small communication penalty. The measurements are also consistent with the analytical mechanisms identified earlier: Eve's degradation increases with the clutter-induced, kurtosis-dependent interference in Lemma~\ref{thm:eve_mse}, whereas Alice experiences only the comparatively mild $\ism$-dependent noise enhancement of Lemma~\ref{thm:alice_mse}. Notably, positive sensing privacy is achieved even though Eve has a $6.7$~dB target-SNR advantage over Alice, demonstrating the effectiveness of the Eve-agnostic design under a challenging link geometry.

\section{Conclusion}\label{sec:conclusion}
We presented sensing-privacy-enhancing constellation shaping for OFDM-ISAC systems by exploiting the receiver asymmetry between a legitimate sensing receiver and a passive eavesdropper. We defined and derived a closed-form sensing-privacy metric governed by the kurtosis $\kurt$ and inverse second-order moment $\ism$, and showed that positive skewness of the symbol-power distribution is a key condition for improving sensing privacy. Based on this analysis, Eve-aware and Eve-agnostic constellation shaping designs were developed to balance sensing privacy and communication reliability. Simulations demonstrated clear privacy--communication trade-offs, while over-the-air experiments showed that the Eve-agnostic design can significantly degrade Eve's ranging performance with only a small impact on Alice and communication throughput. Future work includes joint power--constellation shaping, PAPR-aware designs, and extensions to MIMO to enhance sensing-privacy in deployable ISAC systems.

\appendices

\section{Proof of Lemmas \ref{thm:eve_mse} and \ref{thm:alice_mse}}\label{app:mse}
Let $s(\tau)=\mathbf{h}^H(\tau)\mathbf{y}^T$ and $f(\tau)=|s(\tau)|^2$. A second-order expansion around $\tau_k$ gives
\begin{align}\label{eq:err_general}
\E[(\hat\tau_{k}-\tau_{k})^2] & \approx \frac{\E[|\dot s(\tau_{k})|^2]}{2|\E[\ddot s(\tau_{k})]|^2},\\
|\E[\ddot s(\tau_k)]| & = (2\pi\Delta f)^2|\alpha_k|\sum_{n=0}^{N-1}n^2.
\end{align}
For Eve with MF in the case of an ideal reference,
\begin{equation}\label{eq:eve_mf}
\mathbf{y}_{E,\rm MF} = \mathbf{a}_E^T\mathbf{H}_E|\mathbf{X}|^2 + \mathbf{z}_{E,s}\odot\mathbf{x}^*.
\end{equation}
At the target delay, the derivative fluctuation is determined by the data-dependent sidelobes and receiver noise. Using $\mathrm{Var}(|x_n|^2)=\kurt-1$ and neglecting cross terms between resolvable scatterers,
\begin{equation}
\E[|\dot s_{\rm MF}|^2] = (2\pi\Delta f)^2\Big((\kurt-1)\sum_{j\neq k}^{K_E}|\alpha_{E,j}|^2+\sigma_E^2\Big)\sum_{n=0}^{N-1}n^2 .
\end{equation}
Substitution into \eqref{eq:err_general} yields \eqref{eq:mse_eve}. Finite reference quality only adds nonnegative fluctuation power and therefore gives the stated inequality.

Similarly, for Alice with RF, 
\begin{equation}
\E[|\dot s_{\rm RF}|^2] = (2\pi\Delta f)^2\sigma_A^2 \ism\sum_{n=0}^{N-1}n^2.
\end{equation}
Substituting into \eqref{eq:err_general} yields \eqref{eq:mse_alice}.
\hfill$\blacksquare$

\section{Proof of Proposition \ref{prop:skewness}}\label{app:skewness}
Since $t=|x_n|^2=1+u$ with $\E[u]=0$, we have $\kurt=\E[(1+u)^2]=1+\E[u^2]$. For the inverse second-order moment, using
\begin{equation}\label{eq:recip_identity}
\frac{1}{1+u} = 1 - u + u^2 - u^3 + \frac{u^4}{1+u},
\end{equation}
and taking the expectation gives
\begin{equation}\label{eq:nu_exact}
\ism = 1 + \E[u^2] - \E[u^3] + \E\!\left[\frac{u^4}{t}\right].
\end{equation}
Subtracting the two expressions yields \eqref{eq:exact_gap}. Since $t>0$, $\E[u^4/t]\geq0$, which directly gives \eqref{eq:skew_necessary} and shows that $\E[u^3]>0$ is necessary for a positive moment gap. When fourth- and higher-order fluctuations are negligible, the residual term becomes negligible, giving $\kurt-\ism\approx\E[u^3]$. \hfill$\blacksquare$

\section{Proof of Proposition \ref{prop:tworing}}\label{app:tworing}
Eliminating $P_i$ using \eqref{eq:tworing_power} gives $dP_i/dP_o=-\beta/(1-\beta)$. Differentiating \eqref{eq:tworing_moments},
\begin{align}
\frac{\partial G}{\partial P_o}
&=  \beta(P_o-P_i)\left[2-\frac{P_o+P_i}{P_o^2P_i^2}\right]
\end{align}
which proves \eqref{eq:tworing_grad}.

For (i), let $P_o=1+u_o$ and $P_i=1-\frac{\beta}{1-\beta}u_o$. Applying Proposition~\ref{prop:skewness} and retaining terms up to the third order gives
$G\approx\beta\big(1-\frac{\beta^2}{(1-\beta)^2}\big)u_o^3$, whose coefficient is positive if and only if $\beta<1/2$.

For (ii), as $P_i\to0$, $\ism\geq(1-\beta)/P_i\to\infty$, whereas $\kurt$ remains finite since $P_o\to1/\beta$. Hence, $G\to-\infty$.

For (iii), when $\beta<1/2$, property (i) shows that $G$ initially increases from the PSK point, whereas property (ii) gives $G\to-\infty$ at the opposite boundary. Thus, an interior maximum exists and must satisfy the bracket-nulling condition in \eqref{eq:tworing_grad}, yielding \eqref{eq:tworing_opt}. To establish uniqueness, define $r=P_o/P_i>1$ and
$q(r)=(P_o+P_i)/(P_o^2P_i^2)=(r+1)(1-\beta+\beta r)^3/r^2$. For $\beta<1/2$, $q(r)$ first decreases from $q(1)=2$ and then increases monotonically to infinity, so $q(r)=2$ has exactly one solution for $r>1$. Therefore, the interior stationary point is unique and is the global maximum. \hfill$\blacksquare$

\bibliographystyle{IEEEtran}
\bibliography{IEEEabrv,reference}
\end{document}